\documentclass[12pt]{article}

\usepackage{amsmath}
\usepackage{ucs}
\usepackage[utf8x]{inputenc}
\usepackage{mathtools}

\usepackage{bbm}
\usepackage[pdftex]{graphicx}
\usepackage{makeidx}
\usepackage{amsthm}
\usepackage{fancyhdr}
\usepackage{comment}
\usepackage{tikz}
\usetikzlibrary{arrows,calc,shapes,decorations.pathreplacing}
\usepackage{pgfplots}
\usepackage[hidelinks]{hyperref}
\usepackage{flafter}

\usepackage{graphicx}
\graphicspath{ {./images/} }

\def\P{\mbox{\rm P}}

\usepackage{graphicx}
\usepackage{subcaption}
\usepackage{mwe}

\usepackage[ruled,vlined]{algorithm2e}
\usepackage[margin=1in]{geometry}
\usepackage{tasks,enumerate,xcolor,marginnote}

\colorlet{exercisecolor}{green!70!black!50!brown}
\settasks{counter-format=tsk[a],label-format=\sffamily\bfseries\color{exercisecolor}}
\reversemarginpar
\newtheorem{proposition}{Proposition}

\def\P{\Pi}

\def\P{\mbox{\rm P}}

\def\C{{\bf C}}

\def\P{\mbox{\rm P}}

\newcommand{\ssigma}{\boldsymbol{\Sigma}}

\newcommand{\ttheta}{\boldsymbol{\theta}}

\newcommand{\y}{\boldsymbol{y}}
\newcommand{\be}{\boldsymbol{e}}

\usepackage[font=small,labelfont=bf]{caption}

\newcommand{\Di}{\boldsymbol{D}}

\usepackage[round]{natbib}

\usepackage{authblk}

\title{Estimating Mean Viral Load Trajectory from Intermittent Longitudinal Data and  Unknown Time Origins}

\author[1,2]{Yonatan Woodbridge\thanks{Corresponding Author}}
\author[3]{Micha Mandel}
\author[4]{Yair Goldberg}
\author[1,5]{Amit Huppert}
\affil[1]{The Gertner Institute for Epidemiology \& Health Policy Research, Sheba Medical Center, Ramat Gan, Israel}
\affil[2]{Department of Computer Science, Holon Institute of Technology, Holon, Israel}
\affil[3]{Department of Statistics and Data Science, The Hebrew University of Jerusalem, Jerusalem, Israel}
\affil[4]{Faculty of Industrial Engineering and Management, Technion - Israel Institute of Technology, Haifa, Israel}
\affil[5]{Faculty of Medicine, Tel Aviv University, Israel}
\begin{document}

\clearpage
\maketitle
\thispagestyle{empty}

\newpage

\begin{abstract} Viral load (VL) in the respiratory tract is the leading proxy for assessing infectiousness potential. Understanding the dynamics of disease-related VL within the host is very important and help to determine different policy and health recommendations. However, often only partial followup data are available with unknown infection date. In this paper we introduce a discrete time likelihood-based approach to modeling and estimating partial observed longitudinal samples. We model the VL trajectory by a multivariate normal distribution that accounts for possible correlation between measurements within individuals. We derive an expectation-maximization (EM) algorithm which treats the unknown time origins and the missing measurements as latent variables. 
Our main motivation is the reconstruction of the daily mean SARS-Cov-2 VL, given measurements performed on random patients, whose VL was measured multiple times on different days. The method is applied to SARS-Cov-2 cycle-threshold-value data collected in Israel.
\end{abstract}

Keywords: Ct-value, EM Algorithm, Multivariate Normal Distribution, SARS-Cov-2

\newpage 
\pagestyle{plain}
\setcounter{page}{1}

\section{Introduction}

The Viral load (VL) is the amount of viral nucleic acid within the host, expressed as the number of viral particles in a given volume. The typical measure unit is the cycle-threshold-value (Ct-value). In SARS-Cov-2, the typical Ct-value for a positive sample is between 15 and 40 representing the number of duplications required for the amount of viral genetic material to reach a certain detectable fluorescence threshold. It is inversely correlated to the VL, with a lower Ct-value indicating a higher VL. Sampling is done by taking a nasal swab conducted in testing centers or medical clinics, after which the sample is transferred to one of several special labs for analysis. 

VL typically increases exponentially after infection, until reaching a peak after which it starts to decline exponentially  as a result of the host's immune response \citep{kissler2021viral}. Understanding the VL trajectory is of great interest, as it affects the rate of infectiousness, generation time and disease duration \citep{marks2021transmission}. Real time estimates of quantitative viral shedding dynamics will enable better evidence based public health interventions, such as lock-downs, length of quarantine, mask-wearing and other health-related policies.  Most VL studies are based on longitudinal data, making them difficult and expensive to conduct 
\citep{chia2022virological,ke2022daily,hay2022quantifying}. A second difficulty stems from the fact that for most cases the exact day of infection is unknown or uncertain, due to the study design which follows participants who have been constantly monitored after diagnosis \citep{hay2022quantifying,chia2022virological}. As a practical solution, the time origin can be  defined as the study-entry time, which mostly coincides with symptom appearance. 
In \cite{hay2022quantifying}, the study group was also monitored before infection onset, which allowed a more accurate determination of infection day and VL estimation. Nevertheless, continuous daily monitoring is both expensive and challenging, and statistical methods to estimate the VL typical curve under missing onset date and a partial follow-up are required. Motivated by recent SARS-Cov-2 viral load studies, we use  a large Israeli data set that was collected routinely and did not require a complex surveillance. We focus on discrete-time trajectory reconstruction, given incomplete longitudinal measurements from two or more different time points on each individual. 
Our main goal is to provide basic practical guidelines for data collection and VL curve reconstruction that can be used in future outbreaks.


Several studies addressed the problem of unknown time origin in the context of HIV infections, since initial infection time is typically unknown (the time scale since infection until detection in HIV is months or years while for the case of SARS-Cov-2 it takes several days).  \cite{berman1990stochastic} studied the T4 level  trajectory from infection, 
using a stochastic process with exponential damping function, treating the time from  infection to diagnosis as a latent variable. A similar approach was used in \cite{dubin1994estimation} to estimate the time from infection. Other models use empirical Bayes approaches \citep{degruttola1991modeling},
or biology-based dynamical models 
\citep{drylewicz2010modeling}  to investigate the progression of various biomarkers from initial infection. A recent paper \citep{wang2022time} proposed a likelihood-based estimation method for longitudinal trajectory estimation, which can be incorporated into survival models. 
This model can be applied to a wider range of datasets, as demonstrated on cervical dilation and medfly data.

In this paper, we study a discrete-time longitudinal trajectory estimation method. Similar to \cite{wang2022time}, our method is based on maximum likelihood estimation (MLE). However, here we consider the discrete-time scenario with only few measurements given for each individual.
Our main goal is to make use of data collected in routine surveillance and estimate the mean population trajectory; that is, the mean Ct at each time point, if all individual trajectories were aligned up, starting from a common time origin representing the infection time.
We incorporate in our model a covariance matrix which accounts for possible within-individual correlation, as in longitudinal data analysis or in linear mixed-effects models \citep{myers2012generalized}. Overall,
the unknown parameters included in the likelihood function are the means, covariance matrix, and the  probabilities of the time points at which measurements were performed. 
Using the normal distribution to model the measured values, we derive an EM algorithm for the MLE, which treats the unknown time-origin and the unknown measurements as latent variables. Owing to the discrete-time approach and the multivariate normal model, the EM procedure mostly involves closed form calculations and standard quadratic function minimizations.
The model is flexible and
various constraints can be imposed on the parameters, such as unimodality and structured covariance matrices. We show that under mild conditions, the model is identifiable in all its parameters.



Our model formulation is motivated by SARS-Cov-2 VL studies. We use a dataset on numerous PCR tests and VL measurements conducted by a major Israeli lab, between January to June 2022. Many individuals had their VL measured twice or more on different days. Although the  day of infection is unknown, 
the large number of samples enable reconstructing the population's VL trajectory using our model formulation and assumptions. We compare our results to known results in the literature.

The paper is organized as follows. Section 2 formulates the model, the missing data-generating process, and derives the likelihood function. Section 3 proves that under certain mild conditions, the parameters are identifiable. Section 4 develops an EM algorithm with and without structural constraints on the model's parameters. Section 5 reports results of numerical simulations, as well as applies the method on  SARS-Cov-2 data. Section 6 summarizes the paper and discusses limitations, and future directions.

\section{Model}

Let $x$ be an integer denoting the day following the time origin, and let $y_x$ be the outcome of interest as measured on day $x$. We assume a model of the form $y_x=\theta_x + \epsilon_x$, where $\theta_x$ $(x=1,\ldots,2d-1)$ are fixed parameters, and $(\epsilon_1,...,\epsilon_{2d-1}) \sim F$ are zero-mean random variables which are possibly correlated; the use of vectors of length $2d-1$ will be explained in the sequel. We concentrate on the multivariate normal case $F=\mathcal{N}(0,\ssigma)$, where $\ssigma\in \mathbbm{R}^{2d-1\times 2d-1}$ is a covariance matrix.  The parameters  $\ttheta=(\theta_1,...,\theta_{2d-1})$ and $\ssigma$ are unknown, and should be estimated using incomplete data, as described below.

For $n$ statistically independent samples $(y_{i,1},...,y_{i,2d-1})\sim\mathcal{N}(\ttheta,\ssigma)$,  parameter
estimation by maximum likelihood is straightforward. Suppose now that only $m_i \in \{1,\ldots,d\}$ components from each vector are observed, but their indexing is unknown. Concretely, the data comprise of $n$ vectors of observations $\{(\tilde{y}_{i,1},\ldots,\tilde{y}_{i,m_i})\}_{i=1}^n$ where $\tilde{y}_{i,j}=y_{i,x_{i,j}}$. However, instead of fully observing the indices $(x_{i,1},\ldots,x_{i,m_i})$, we assume that $x_{i,1}\in\{1,\ldots,d\}$ and that only the difference between the indices $\Delta_{i,j}$ are available:
\begin{equation*}
    \Delta_{i,j} = x_{i,j+1}-x_{i,j},
\end{equation*}
for $j=1,\ldots,m_i-1$. For our motivating example, the assumptions state that the first observation is during the ``infection" period, but the exact day post infection is unknown, therefore only the differences between measurement times are exactly observed. If $m_i=1$ we formally define $\Delta_{i,j}=0$. Otherwise,
since the $x_{i,j}$'s are integers, $\Delta_{i,j}$'s are also positive integers satisfying $\sum_j\Delta_{i,j}=x_{i,m_i}-x_{i,1}<d$. 
The goal is to estimate $\ttheta$ and $\ssigma$ based on  $\{(\tilde y_{i,1},\ldots,\tilde y_{i,m_i}, \Delta_{i,1},\ldots,\Delta_{i,m_i-1})\}_{i=1,..,n}$.

The model considers $d$ possible time points during which the infection is  ``active". However, in practice $d$ is often unknown, and the latest observation time $x_{i,m_i}$ might be larger than $d$. In fact, as $x_{i,1}=d$ and $x_{i,m_i}-x_{i,1}=d-1$ are both possible, $x_{i,m_i}$ can be as large as $2d-1$. In typical settings, such as that of VL, it is reasonable to assume a ``steady state" after time $d$, so we propose the following convention:
The means and standard deviations from time $d$ onward are all equal, that is, $\theta_j=\theta_{d}$ and $\sigma^2_j=\sigma^2_{d}$ for $j=d,d+1,\ldots,2d-1$.
As we will see in the numerical simulations, selecting a proper $d$ is important for the estimate's accuracy.


To construct the likelihood, let $q_j= p(x_{i,1}=j)$
be the probability of the first observed time $x_{i,1}$ supported on $1,...,d$. Assume a non-informative selection mechanism in a sense that given $(x_{i,1},\ldots,x_{i,m_i})$, the observed data $(\tilde y_{i,1},\ldots,\tilde y_{i,m_i})$ follows the model $y_{x_j}=\theta_{x_j} + \epsilon_{x_j}$ as described above (i.e, there is no selection bias). While this assumption is strong, it holds in a controlled experiment, which can be conducted during a pandemic in order to support policy decisions. We comment more on this in the Discussion section.

For ease of notation, we use a single index $j$ to denote the variance parameters (i.e., the diagonal values of $\ssigma$), such that:
\begin{equation}
    \text{diag} (\ssigma)= (\sigma_1^2,...,\sigma_{2d-1}^2).
    \label{diagsigma}
\end{equation}
For the off diagonal covariance terms, we use two indices: $[\ssigma]_{k,\ell}=\sigma_{k\ell}$.
Given $x_{i,1}=j$, the vector $\tilde{\y}_i = (\tilde y_{i,1},\ldots,\tilde y_{i,m_i})$ has an $m_i$-variate normal distribution with mean $\ttheta_{i|j}=(\theta_j,\theta_{j+\Delta_{i,1}},\ldots,\theta_{j+\Delta_{i,m_i-1}})$, and covariance
\begin{equation}
    \ssigma_{i|j}=
    \begin{pmatrix}
        \sigma_{j}^2 & \cdots & 
        \sigma_{j,j+\Delta_{i,m_i-1}} \\
        \vdots & \ddots & \vdots \\
        \sigma_{j,j+\Delta_{i,m_i-1}} & \cdots & \sigma_{j+\Delta_{i,m_i-1}}^2
    \end{pmatrix}.
    \label{submatrix1}
\end{equation}
As $x_{i,1}$ is not observed, the likelihood is a mixture of normal densities:
\begin{equation}
   \mathcal{L}_m=
    \prod\limits_{i=1}^n
    \sum\limits_{j=1}^{d}
    \frac{q_j}{(2\pi)^{m_i/2} \:
    \sqrt{|\ssigma_{i|j}|}\:
    }
    \text{exp}
    \bigg{(}
    -\frac{1}{2}
    (\tilde{\y}_i -\ttheta_{i|j})^T
    \ssigma_{i|j}^{-1}
   (\tilde{\y}_i -\ttheta_{i|j})
    \bigg{)}.
\label{LM}
\end{equation}
Direct maximization of \eqref{LM} is difficult. However, treating $x_{i,1}$ and the unobserved components within $(y_{i,1},...,y_{i,2d-1})$  as latent variables, we derive in the next section an EM algorithm  \citep{dempster1977maximum} that is composed of standard calculation steps.

\section{Model Identification}
In this section, we focus on the case where only pairs of samples $(\tilde{y}_{i,1},\tilde{y}_{i,2})$ are given, together with the difference $\Delta_{i1}$. We prove that under certain uniqueness conditions, the model is identifiable. As a matter of convenience, we consider $i=1$ and omit the indices, and use  
the notation $\mathcal{N}(\cdot;\cdot)$ to mark the marginal bivariate Normal p.d.f with given parameters. Thus, for a given $\Delta$, the joint mixture distribution is
\begin{equation}
    \mathcal{L}_{\ttheta,\ssigma,\boldsymbol{q}}(\tilde{y}_{1},\tilde{y}_{2},\Delta
    )=
    \sum\limits_{j=1}^{d} 
    q_j 
    \mathcal{N} \Big{(}
    \tilde{y}_{1}, \tilde{y}_{2}
    ;
    \ttheta_{1|j}
    ,\ssigma_{1|j}
     \Big{)}.
     \label{Gaussform}
\end{equation} 
The mean and covariance $\ttheta_{1|j},\ssigma_{1|j}$ are defined in Section 2 (see \eqref{submatrix1} and the discussion above it). 
To prove model identifiability, we show that if $\mathcal{L}_{\ttheta,\ssigma,\boldsymbol{q}}(\tilde{y}_{1},\tilde{y}_{2},\Delta)=\mathcal{L}_{\ttheta^{'},\ssigma^{'},\boldsymbol{q}^{'}}(\tilde{y}_{1},\tilde{y}_{2},\Delta)$ for all $\tilde{y}_{1},\tilde{y}_{2} \in \mathbbm{R}$ and $\Delta\in \{1,...,d\}$, then  $\ttheta=\ttheta^{'}$,
$\ssigma=\ssigma^{'}$, and
$\boldsymbol{q}=\boldsymbol{q}^{'}$.

Without any constraints, the model is not identifiable. For example, if $q_1=0$ then clearly $\theta_1$ and $\sigma^2_1$ are not identifiable. Similarly, if with probability 1 $\Delta\ne j$ for some $1\le j \le d-1$, then the $(k,k+j)$th ($k=1,\ldots,d-j$) elements of $\ssigma$ are not identifiable.
However, the model can be partially identifiable if we assume the following reasonable conditions:
\begin{itemize}
    \item \textbf{Condition 1}. 
    The pairs $(\theta_j,\sigma^2_j)$ $j=1,\ldots,d$
are unique.
\item \textbf{Condition 2}. The mixing probabilities $q_j$ are positive for all $j=1,...,d$.
\end{itemize}

The proof uses results on identifiability of finite mixture models  \citep{teicher1963identifiability,yakowitz1968identifiability}. 
Specifically, \cite{yakowitz1968identifiability} proved that an $n$-dimensional Gaussian mixture model (GMM) is identifiable up to permutations of indices. Taking the special 2-dimensional case, \cite{yakowitz1968identifiability} shows that if a GMM p.d.f with $d$ different Gaussian components is
identical to another, such that
$$\sum\limits_{j=1}^d \alpha_j \mathcal{N}(x;\boldsymbol{\mu}_j,\boldsymbol{Q}_j)=\sum\limits_{j=1}^d \alpha'_j \mathcal{N}(x;\boldsymbol{\mu}'_j,\boldsymbol{Q}'_j), \quad \forall x\in\mathbbm{R}^2,$$
then a permutation $\pi$ exists satisfying $\pi(\alpha'_1,...,\alpha'_d)=(\alpha_1,...,\alpha_d)$, 
$\pi(\boldsymbol{\mu}'_1,...,\boldsymbol{\mu}'_d)=(\boldsymbol{\mu}_1,...,\boldsymbol{\mu}_d)$ and $\pi(\boldsymbol{Q}'_1,...,\boldsymbol{Q}'_d)=(\boldsymbol{Q}_1,...,\boldsymbol{Q}_d)$. 

Our model does not fall into the GMM framework, as only bivariate marginals of the $2d-1$ dimensional normal distribution are observed. Nevertheless, we can use the identifiability result of GMM to prove identification in our setting.

\begin{proposition} \label{prop:iden}
Suppose that $P(\Delta=j)>0$ for  $j=1,\ldots,d-1$. Under conditions 1-2,
model
\eqref{Gaussform}  is identifiable in  $\ttheta,\ssigma$ and $\boldsymbol{q}$.
\end{proposition}

\begin{proof}

For $\Delta=1$, we observe the mixture
\begin{equation} \label{eq:mixbivariate}
    \sum\limits_{j=1}^{d} 
    q_j 
    \mathcal{N} \Big{(}
     \tilde{y}_{1},
     \tilde{y}_{1}
    ;
    \begin{psmallmatrix}
     \theta_j \\ 
     \theta_{j+1}
    \end{psmallmatrix},
    \begin{psmallmatrix}
     \sigma^2_{j} & \sigma_{j,j+1} \\
     \sigma_{j,j+1}
     & \sigma^2_{j+1}
    \end{psmallmatrix}
     \Big{)}.
\end{equation}
By identifiability of GMM models, the bivariate distributions in \eqref{eq:mixbivariate} and the mixing probabilities $q_1,\ldots,q_d$ are all identifiable up to an indices permutation. By the first condition, the pair $(\theta_1,\sigma^2_1)$ appears only in one of the mixtures, which identifies the first index. This also identifies $q_1$, $(\theta_2,\sigma^2_2)$ and $\sigma_{1,2}$ via identification of $\mathcal{N} \Big{(}
     \tilde{y}_{1}, \tilde{y}_{1} ;\begin{psmallmatrix}
     \theta_1 \\  \theta_{2}
    \end{psmallmatrix}, \begin{psmallmatrix}
     \sigma^2_{1} & \sigma_{1,2} \\
     \sigma_{1,2}  & \sigma^1_{2}
    \end{psmallmatrix}
     \Big{)}.$ 
Once $(\theta_2,\sigma^2_2)$ is identified, also $(\theta_3,\sigma^2_3)$ and $\sigma_{2,3}$ are identified via identification of $\mathcal{N} \Big{(}
     \tilde{y}_{1}, \tilde{y}_{1} ;\begin{psmallmatrix}
     \theta_2 \\  \theta_{3}
    \end{psmallmatrix}, \begin{psmallmatrix}
     \sigma^2_{2} & \sigma_{2,3} \\
     \sigma_{2,3}  & \sigma^1_{3}
    \end{psmallmatrix}
     \Big{)}.$  
Continuing with the same reasoning shows identifiability of the mixing probabilities and all means, variances and single-lag correlations. Repeating the same arguments for $\Delta=k$ ($2\le k\le d-1$) establishes identification of the $k$-lag correlations.
\end{proof}

\section{An EM algorithm}

\subsection{The unconstrained model}
In terms of an EM algorithm, we define the complete likelihood to be of the time origin $x_{i,1}$ and all measurements $\y_i=(y_{i,1},\ldots,y_{i,2d-1})^T$ of individual $i$. 
Since $\y_i\sim\mathcal{N}(\ttheta,\ssigma)$ $i=1,\ldots,n$ are independent,
the complete likelihood simply becomes:
\begin{equation}
    \mathcal{L}_c = 
    \prod\limits_{i=1}^n
    \prod\limits_{j=1}^{d}
    \bigg{[}    \frac{q_j}{(2\pi)^{(2d-1)/2} \:
    \sqrt{|\ssigma|}
    }
    \text{exp}
    \bigg{(}
   -
   \frac{1}{2}
   (\y_i-\ttheta)^T 
   \ssigma^{-1}
(\y_i-\ttheta)
    \bigg{)}
    \bigg{]}^{\mathbbm{1}_{x_{i,1}=j}} 
    \label{Lc},
\end{equation}
where ${\mathbbm{1}_{x_{i,1}=j}}$ indicates the event $\{x_{i,1}=j\}$. The conditional expectation of  $-\text{log} \mathcal{L}_c$ is calculated in the E-step, which is then maximized over the unknown parameters in the M-step. These steps are repeated until convergence.

\textbf{E-step}.
Denote by ${\Theta}^{(t)}$
the set of estimated parameters obtained at the $t$-th iteration, which contains the mean vector ${\ttheta}^{(t)}$, the covariance ${\ssigma}^{(t)}$, and the prior probability parameters $\boldsymbol{q}^{(t)}=\big{\{}
q_1^{(t)},...,q_{d}^{(t)}\big{\}}$.
For ease of notation,
we use $D_i$ 
to mark the observed data on subject $i$, i.e $\{D_i\}_{i=1\ldots,n}\equiv\{(\tilde y_{i,1},\ldots,\tilde y_{i,m_i}, \Delta_{i,1},\ldots,\Delta_{i,m_i-1})\}_{i=1,..,n}$,
and the boldface letter $\boldsymbol{D}=\cup_iD_i$
to mark the entire observed data.

The negative log of \eqref{Lc}, without constant terms, can be written as:
\begin{equation}
    -\text{log}
    \Big{(}
    \mathcal{L}_c
    \Big{)}=
    \sum\limits_{i=1}^n
    \sum\limits_{j=1}^{d}
    \mathbbm{1}_{x_{i,1}=j} 
    \Bigg{\{}
    \frac{1}{2}\text{log} (|\ssigma|) +
    \frac{1}{2}
    (\y_i-\ttheta)^T 
   \ssigma^{-1}
(\y_i-\ttheta)
-\text{log}\: (q_j)
    \Bigg{\}}.
    \label{logelc}
\end{equation}
Computing the conditional expectation of \eqref{logelc} requires the calculation of $\mathbbm{E}_{{\Theta}^{(t)}}
    \Big{(}
    \mathbbm{
    1}_{x_{i,1}=j
    }
    \Big{|}
    \boldsymbol{D}
    \Big{)}$ and 
$\mathbbm{E}_{{\Theta}^{(t)}}
    \Big{(}
    \mathbbm{
    1}_{x_{i,1}=j
    }  f(\y_i)
    \Big{|}
   \boldsymbol{D}
    \Big{)}$ for a linear and a quadratic function  $f(\cdot)$.
Denoting the former as $E_{ij}^{(t)}$, we have: 
\begin{equation*}
    E^{(t)}_{ij}  = p_{\Theta^{(t)}}(x_{i,1}=j \mid \boldsymbol{D}) =
    \frac{p_{\Theta^{(t)}}(D_i \mid x_{i,1}=j) q_j^{(t)}}
    { \sum_{k=1}^{d}  p_{\Theta^{(t)}}(D_i \mid x_{i,1}=k)q_k^{(t)} }.
\end{equation*}
The expectation $E^{(t)}_{ij}$ is readily obtained by recalling that $D_i \mid \{x_{i,1}=k\}$ has a normal distribution with mean $\theta_{i|k}$ and covariance matrix $\Sigma_{i|k}$; see \eqref{submatrix1} and the discussion above it.

For the second expectation,  
$\mathbbm{E}_{{\Theta}^{(t)}}\Big{(}\mathbbm{ 1}_{x_{i,1}=j} f(\y_i)\mid\boldsymbol{D}\Big{)}$, we have:
\begin{equation*}
    \mathbbm{E}_{{\Theta}^{(t)}}\Big{(}\mathbbm{1}_{x_{i,1}=j}f(\y_i)
    \mid \boldsymbol{D}\Big{)}=
    \mathbbm{E}_{{\Theta}^{(t)}}\Big{(}f(\y_i)
    \mid x_{i,1}=j,D_i\Big{)} \times E^{(t)}_{ij}.
\end{equation*}    
The expectation $\mathbbm{E}_{{\Theta}^{(t)}}\Big{(}f(\y_i)
    \mid x_{i,1}=j,D_i\Big{)}$ involves calculation of the first two moments of $\y_i$ conditionally on the event $\{D_i,x_{i,1}=j\}$, which can be obtained using known properties of multivariate normal distributions. The details are deferred to Appendix I.

\textbf{M-step}. The M-step should minimize for $\ttheta,\ssigma$, and $q_1,\ldots,q_d$ the expression
\begin{equation}
    -\mathbbm{E}_{\Theta^{(t)}}
    \bigg{(}\text{log}(\mathcal{L}_c)\mid \Di\bigg{)} =
    \sum\limits_{i=1}^n
    \sum\limits_{j=1}^{d}
  E_{ij}^{(t)}
    \Big{(}
    \frac{1}{2}\text{log}(|\ssigma|)-\text{log}(q_j)
    +\frac{1}{2} 
    {\rm trace} \Big{(} \ssigma^{-1} \boldsymbol{Y}_{\ttheta,i|j}^{(t)}
    \Big{)} \Big{)},
    \label{eloglc}
\end{equation}    
where 
\begin{equation}
    \boldsymbol{Y}_{\ttheta,i|j}^{(t)}=
       \ttheta\ttheta^T    -
    \y^{(t)}_{i|j} \:
    \ttheta^T   -
    \ttheta \:(\y^{(t)}_{i|j})^T    +
   \C_{i|j}^{(t)},
   \label{yy}
\end{equation}
and $\y^{(t)}_{i|j}$ and $\C_{i|j}^{(t)}$ are defined in Appendix I.
Without any constraint, the minimum points  $\ttheta^{(t+1)}$ and $\ssigma^{(t+1)}$ of \eqref{eloglc} over $\ttheta$ and $\ssigma$ are similar to the sample mean and covariance. To see that, note that the $\ell$-th partial derivative of \eqref{eloglc} over $\ttheta$, in a trace form, is:
\begin{equation*}
\frac{1}{2}
\sum\limits_{i=1}^n
\sum\limits_{j=1}^d
E_{ij}^{(t)}
\text{trace} 
\Bigg{(}
\ssigma^{-1}
\bigg{(}
\ttheta 
\boldsymbol{e}_{\ell}^T +
 \boldsymbol{e}_{\ell} \ttheta^T
 -\y_{i|j}^{(t)}\boldsymbol{e}_{\ell}^T-\boldsymbol{e}_{\ell}
 (\y_{i|j}^{(t)})^T
\bigg{)}
\Bigg{)},
\end{equation*}
where $\boldsymbol{e}_{\ell}$ is the $\ell$-th unit basis vector.
It can be shown that the gradient becomes zero at the point
\begin{equation}
    \ttheta^{(t+1)}=
    \frac{1}{n}
    \sum\limits_{i=1}^n
    \sum\limits_{j=1}^{d}
    E_{ij}^{(t)}\boldsymbol{y}_{i|j}^{(t)},
    \label{ttheta}
\end{equation}
As for the covariance, the derivative of \eqref{eloglc} over $\ssigma^{-1}$ is:
\begin{equation*}
\frac{1}{2}
\sum\limits_{i=1}^n
\sum\limits_{j=1}^d
E_{ij}^{(t)}
\bigg{(}
\ssigma
-\boldsymbol{Y}_{\ttheta,i|j}^{(t)}
\bigg{)},
\end{equation*}
which becomes zero at
\begin{equation*}
\ssigma^{(t+1)}=
\frac{1}{n}
   \sum\limits_{i=1}^n
    \sum\limits_{j=1}^{d}
E_{ij}^{(t)}\boldsymbol{Y}_{\ttheta^{(t+1)},i|j}^{(t)}.
\end{equation*}
Lastly, the updated probability parameters are the sample proportions:
\begin{equation*}
    {q}^{(t+1)}_k = \frac{
    \sum_{i=1}^n E_{ik}^{(t)}
    }
    {
    \sum_{i=1}^n \sum_{j=1}^{d} E_{ij}^{(t)}
    } \quad k=1,...,d\;.
\end{equation*}

\subsection{Model constraints} \label{sec:constraint}
Since the model has many unknown parameters, it is helpful to specify structural constraints that are based on prior knowledge and are relevant to the problem's domain.
We discuss several constraints that are relevant to our motivation problem of VL reconstruction.

\textbf{Means}.
First note from \eqref{eloglc} and \eqref{yy} that estimating $\ttheta$ requires solving the following quadratic program:
\begin{equation}
    \min_{\ttheta} \:
    \left\{n\ttheta^T \ssigma^{-1} \ttheta -2\ttheta^T \ssigma^{-1}
    \Big{(}
    \sum\limits_{i=1}^n
    \sum\limits_{j=1}^{d}
E_{ij}^{(t)}\y_{i|j}^{(t)} \Big{)} \right\}.
\label{quadratic}
\end{equation}

We have already made a structural constraint in the model formulation (Section 2), by assuming all means $\theta_{j}$ for $j\ge d$  are equal to $\theta_d$. This reflects the fact that VL stabilizes after recovery. A natural asumption is a unimodal model, as the VL is expected to increase initially until reaching a maximal point, and then to decrease over time. Concretely, assuming a peak at time $d_{max}\in\{1,...,d\}$, we impose the linear constraints $\theta_i\leq \theta_{i+1} \quad \forall i=1,...,d_{max}-1$ and
$\theta_i\geq \theta_{i+1} \quad \forall i=d_{max},...,2d-1$. 
If $d_{max}$ is unknown, we propose to run the EM-algorithm for every $d_{max}=1,...,d$ and select the estimate with the largest
likelihood value.

To further reduce the number of parameters, we focus on a family of unimodal functions specified by three parameters $\alpha_1,\alpha_2,\alpha_3$:
\begin{equation}
    \theta_k=
    \alpha_1 k^{\alpha_2-1}e^{-k/\alpha_3}, \quad 
    \alpha_1,\alpha_2,\alpha_3>0; \quad k=1,\ldots,d,
    \label{thetax}
\end{equation}
and $\theta_k=\theta_d$ for $k\ge d$. These unimodal functions, which we refer to as the unimodal Gamma (due to their similarity to the Gamma function), 
can smoothly capture the exponential growth, followed by an exponential decline rate, which characterizes most Ct-value and VL trajectories \citep{ke2022daily}. 
In that case,
minimization of \eqref{quadratic} can be carried out using a grid search over
$\alpha_1,\alpha_2$ and $\alpha_3$.


\textbf{Covariance}. 
Within individual dependence should be taken into account when samples contain longitudinal measurements per individual. It is natural to assume a model of the form $\ssigma=\sigma^2R$, where $\sigma^2$ is the variance and $R$ the correlation matrix. Various models for the covariance matrix have been suggested, see for example Chapter 7 of \cite{fitzmaurice2012applied}. Since correlation tends to vanish as time difference increases, we choose to focus on a parsimonious first-order autoregressive model $\ssigma=\sigma^2R(\rho)$ where ${\rm Corr}(Y_{ij},Y_{ik})=\rho^{|k-j|}$.
The M-step now involves the minimization of \eqref{eloglc} over $\ssigma$ under the constraint $\ssigma=\sigma^2 R(\rho)$, which leads to:
\begin{equation*}
    \min_{\sigma^2,\rho} \:\:
    (2d-1)n
    \log \sigma^2 
    +
    n\log R(\rho)
    + \frac{1}{\sigma^2} 
    \text{trace}
    \Bigg{(}
    R^{-1}(\rho)
    \bigg{(}
    \sum_{i=1}^n
    \sum_{j=1}^d 
    E_{i|j}^{(t)} 
    \boldsymbol{Y}^{(t)}_{\ttheta,i|j}
    \bigg{)}
    \Bigg{)}
\end{equation*}
For a fixed $\rho_0\in (-1,1)$, simple differentiation shows that  
$$\sigma^2_0=
\frac{1}{(2d-1)n}
\text{trace}
    \Bigg{(}
    R^{-1}(\rho_0)
    \bigg{(}
    \sum_{i=1}^n
    \sum_{j=1}^d 
    E_{i|j}^{(t)} 
    \boldsymbol{Y}^{(t)}_{\ttheta,i|j}
    \bigg{)}
    \Bigg{)}$$
minimizes the above function.  A line search procedure can be applied in order to find the minimum point by calculating $(\rho_0,\sigma_0)$ in a grid over $(-1,1)$, and selecting the point 
that gives the smallest objective value.

Another family of covariance matrices of interest has the linear form
\begin{equation}
    \ssigma(\boldsymbol{\beta}) = \sum_{j=1}^J \beta_j B_j,
\label{mixedeffect}
\end{equation}
where $B_j$ ($j=1,\ldots,J$) are known symmetric matrices. An example  is the heteroskedastic model in which ${\rm Var}(Y_{ij})=\sigma^2_j$. Here $B_j$ for $j=1'\ldots,d$ is a matrix with 1 in the $j$th entry of the diagonal and 0 in all other entries. If in addition ${\rm Cov}(Y_{ij},Y_{ik})=\sigma_{|k-j|}$ depends on the lag between observations, matrices having 1 in entries $(j,j\pm k)$ and 0 otherwise can be added to \eqref{mixedeffect}.
A fixed point procedure for estimating the coefficient vector $\boldsymbol{\beta}$ of a linear covariance model was proposed in \cite{anderson1973asymptotically}; we briefly describe the implementation of this algorithm to our problem in Appendix II. 

\subsection{Technical Remarks}

Since the Ct-value is inversely correlated to the VL, and it is upper bounded by 40, we use the negative linear transformation $y_{ij}=40-{\rm Ct}_{ij}$, where ${\rm Ct}_{ij}$ is the Ct-value of subject $i$ measured on day $j$. The values are highly variable and together with non-convexity of the likelihood can lead to poor estimation. Reducing the number of parameters by using prior knowledge on the expected Ct curve, and specifically focusing on the family of unimodal functions \eqref{thetax}, helps solving the problem to some extent. 

To further deal with the non-convexity of the likelihood, we apply the following procedure. First, we run the algorithm several times, starting from random initializations of all parameters. Second, for each set of initial values, instead of directly estimating the parameters under model \eqref{thetax}, we found it better to first estimate $\ttheta$ under a unimodal constraint, as discussed in Section \ref{sec:constraint}, and use the results as initial values after fitting the function \eqref{thetax} to the estimated means. Specifically, in our data studies we generated five sets of estimates of $(\ttheta,\ssigma,\boldsymbol{q})$ using the EM-algorithm under a unimodal constraint on $\ttheta$ assuming $d_{max}=m$ for $2\leq m \leq 6$, which is its expected range, and then use the results as initial values for the EM algorithm to the unimodal Gamma model  \eqref{thetax}, choosing the estimate that gives the largest likelihood value. The final estimate is the one that maximizes the likelihood over all replications.

\section{Numerical simulations}

\subsection{Performance of the Method}

\cite{hay2022quantifying} measured the Ct-values continuously on a group of NBA players, providing a small sample of real world Ct data. We use this dataset as a baseline for our simulation study. 
For our first simulation setting, we approximated the sample mean of the daily Ct-value by fitting \eqref{thetax}.  
We constructed an
AR(1) covariance matrix by setting $\sigma^2={10},\rho=0.9$ and $\ssigma=\sigma^2 R(\rho)$; 
see discussion about covariance in Section 4.1. Lastly, we defined a vector of decreasing values $\boldsymbol{q}$ in the simplex. We set these parameters as the ground truth when generating the samples and calculating the performance of the method. The parameters $\ttheta$ and $\boldsymbol{q}$ that were used in the simulation are shown as the solid red lines in Figure \ref{fig:synth1}(a) and \ref{fig:synth1}(c), respectively. 

The choice of $\Delta$ is important for the performance of the estimator and can be determined by policy makers in future pandemic-like scenarios. We assume only paired data ($m_i=2$ for all $i$), and compare three distributions for the $\Delta$ values: uniformly distributed over $\{1,...,14\}$, uniformly distributed over $\{2,3\}$, and a decreasing distribution shown by the histogram in Figure \ref{fig:synth1}(c).
To generate a sample, we randomly select integer $1\leq x_{i,1}\leq 14$ using $\boldsymbol{q}$, and  sample $\Delta$ according to its distribution. We then generate $\tilde{y}_{i,1},\tilde{y}_{i,2}$ from a bivariate normal distribution  parameterized by the corresponding means in $\ttheta$ and the $2\times 2$ sub-covariance of $\ssigma$. A scatter plot of the samples $(x_{i,1},\tilde y_{i,1}),(x_{i,2},\tilde y_{i,2})$  is shown in
Figure \ref{fig:synth1}(a).

\begin{figure}[t]
  \begin{minipage}{\textwidth}
    \centering
    \subcaptionbox{synthetic samples}
    {\includegraphics[width=0.45\textwidth]{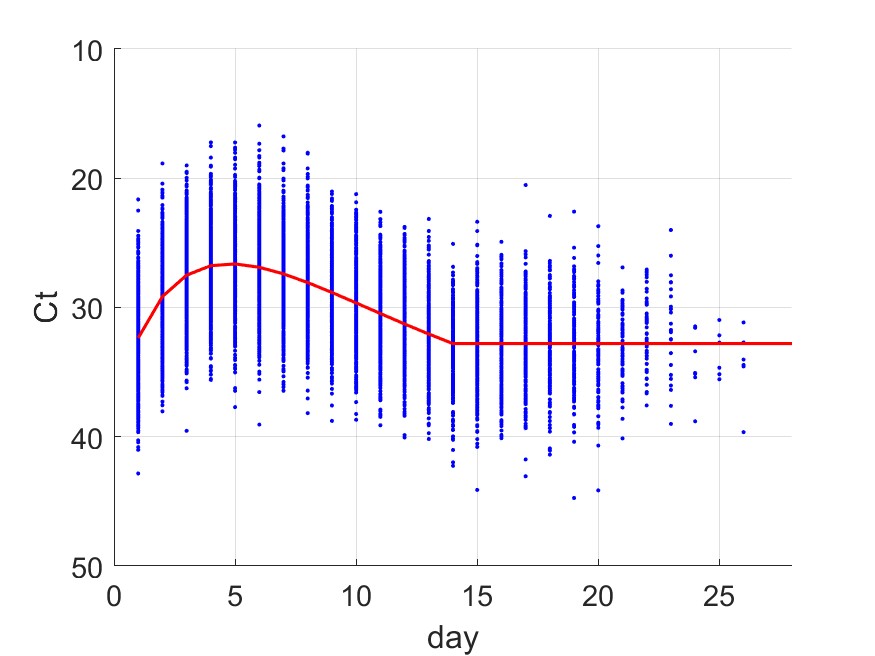}}\quad
    \subcaptionbox{semi-synthetic samples}
     {\includegraphics[width=0.45\textwidth]{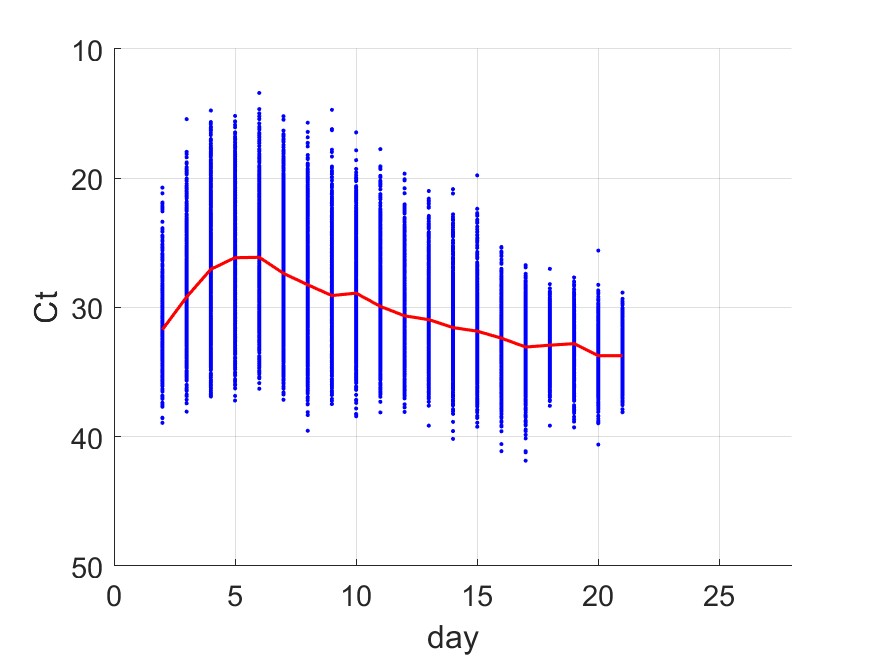}}
     \\
     \centering
    \subcaptionbox{synthetic parameters}
    {\includegraphics[width=0.45\textwidth]{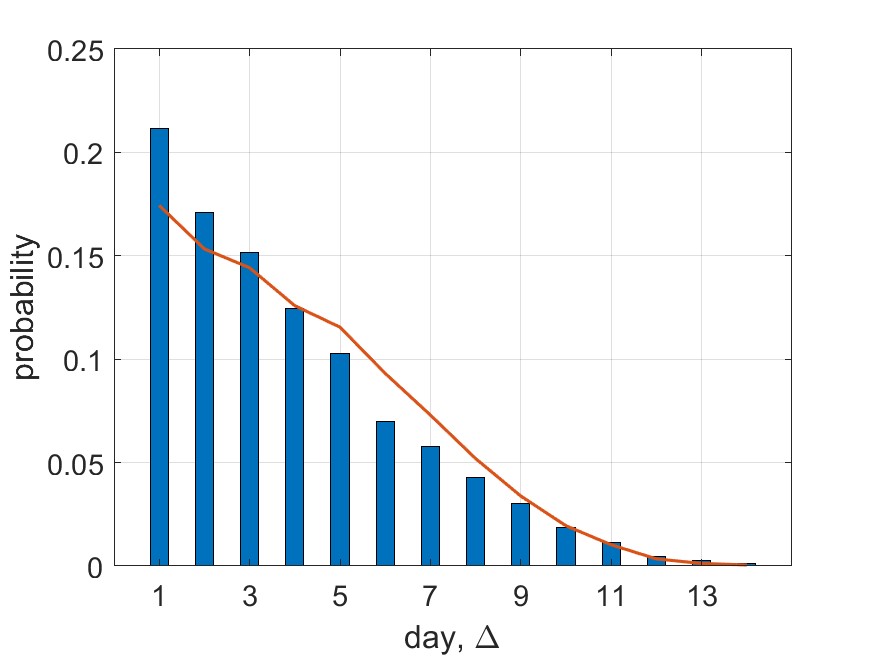}}\quad
    \subcaptionbox{semi-synthetic distributions}
     {\includegraphics[width=0.45\textwidth]{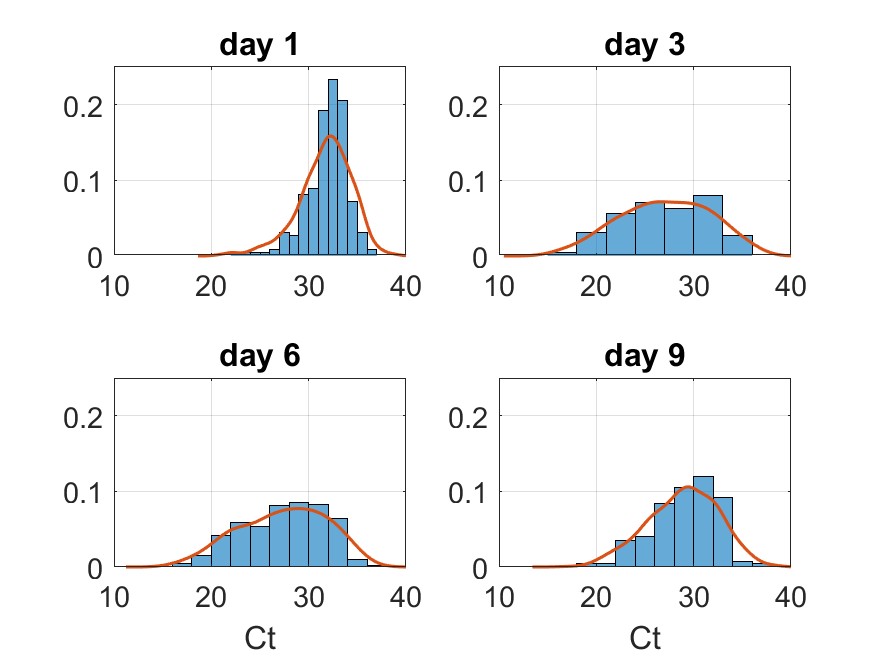}}
    \caption[]{Top row: Ct-value plotted against  day after infection. Blue dots are samples generated synthetically (a) and semi-synthetically (b). Solid red lines represent the daily mean Ct-values. 
    (c): Synthetic sample parameters: $\boldsymbol{q}$ (red line) and $p(\Delta)$ (histogram). (d): Histograms of the Ct-values as provided in \cite{hay2022quantifying} with the density estimation of the
    corresponding
    semi-synthetic Ct-value  distributions (solid red lines), on days 1,3,6 and 9.}
    \label{fig:synth1}
  \end{minipage}
\end{figure}

In order to better mimic the real world Ct-values and to study the performance of the estimator when the assumptions do not hold, we conducted an additional set of simulations, this time sampling directly from the data provided by \cite{hay2022quantifying}. As mentioned before, the viral load of each participant was monitored on a daily basis, so information on the day of infection for each infected participant is quite accurate, and the Ct-values on almost all days following infection are known. As the number of NBA players participating in the study was small, basing the simulation on  individuals' data was impractical. Instead, we sampled independent Ct-values for 20 days using the marginal empirical distributions of the Ct-value data and added a subject-specific $\mathcal{N}(0,1)$ random number to the whole vector in order to add within individual correlations. To further deal with the small sample size, we added for each individual an additional noise for each coordinate using independent $\mathcal{N}(0,1)$ random numbers. The empirical distributions of the Ct-values on days $1,3,6$ and $9$, as well as the corresponding marginal distributions used in the simulation, are shown in Figure \ref{fig:synth1}(d).
After sampling the whole vector, we randomly select two entries from the distribution of $\Delta$ presented in  Figure \ref{fig:synth1}(c). Figure \ref{fig:synth1}(b) is the equivalence of Figure \ref{fig:synth1}(a) presenting the mean and the sampled data for this scenario. Of course, when estimating the parameters, only the data $(\tilde{y}_{i,1},\tilde{y}_{i,2},\Delta_i=x_{i,2}-x_{i,1})$ are used.

We quantify the estimates' accuracy using the normalized mean-square error (NMSE):
$$\text{NMSE}(\hat{\ttheta})=
    \|\hat{\ttheta}-\ttheta\|^2
    / 
    \|\ttheta\|^2.$$ 

We found that the unique-variance AR(1) structure for the covariance does not hold in that case, since the data is heteroscedastic. Indeed, simulations showed  no consistency of the estimate when assuming that structure.
We therefore chose to use the linear covariance model, 
where the main diagonal can have different values, thus accounting for heteroscedasticity. 
To account for within individual correlation we set the two closest off diagonal entries in $\ssigma$ to be non-zero with a unique value for each off-diagonal ($\ssigma_{i,j}=\sigma^2_i$ if $i=j$, $\ssigma_{i,j}={\rm c}_1$ if $|i-j|=1$, $\ssigma_{i,j}={\rm c}_2$ if $|i-j|=2$, and $\ssigma_{i,j}=0$ if $|i-j|\ge 3$).

\begin{figure}[h]
  \begin{minipage}{\textwidth}
    \centering
    \subcaptionbox{$\ttheta$}
    {\includegraphics[width=0.45\textwidth]{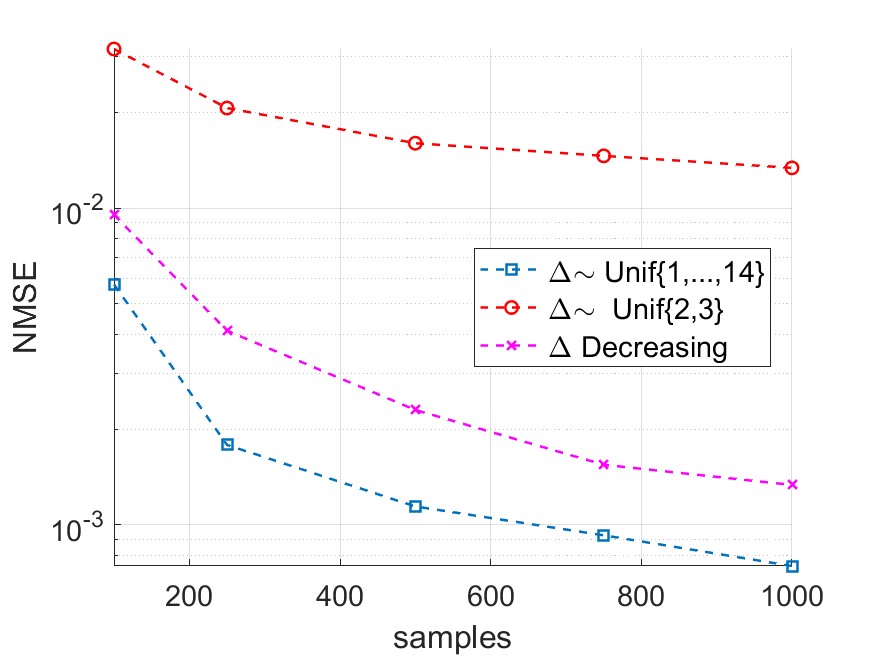}}\quad
    \subcaptionbox{$\sigma,\rho$}
     {\includegraphics[width=0.45\textwidth]{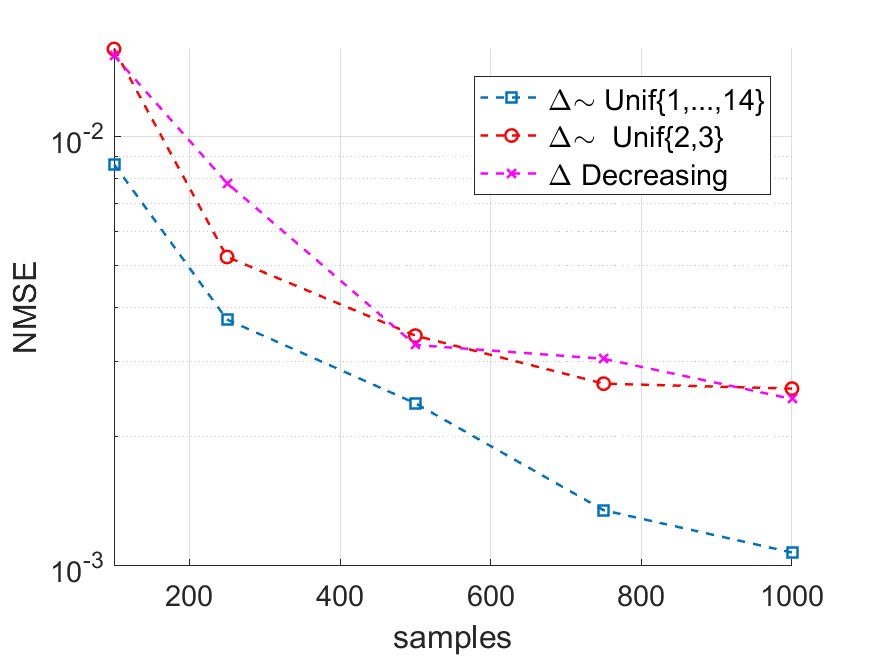}}
     \\
     \centering
    \subcaptionbox{$\boldsymbol{q}$}
    {\includegraphics[width=0.45\textwidth]{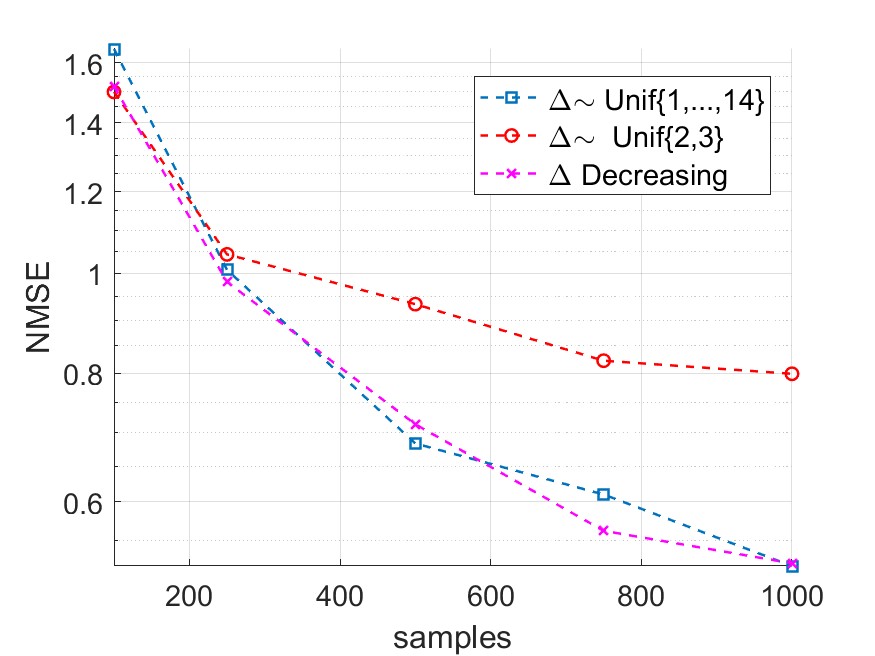}}\quad
    \subcaptionbox{plots of estimates}
     {\includegraphics[width=0.45\textwidth]{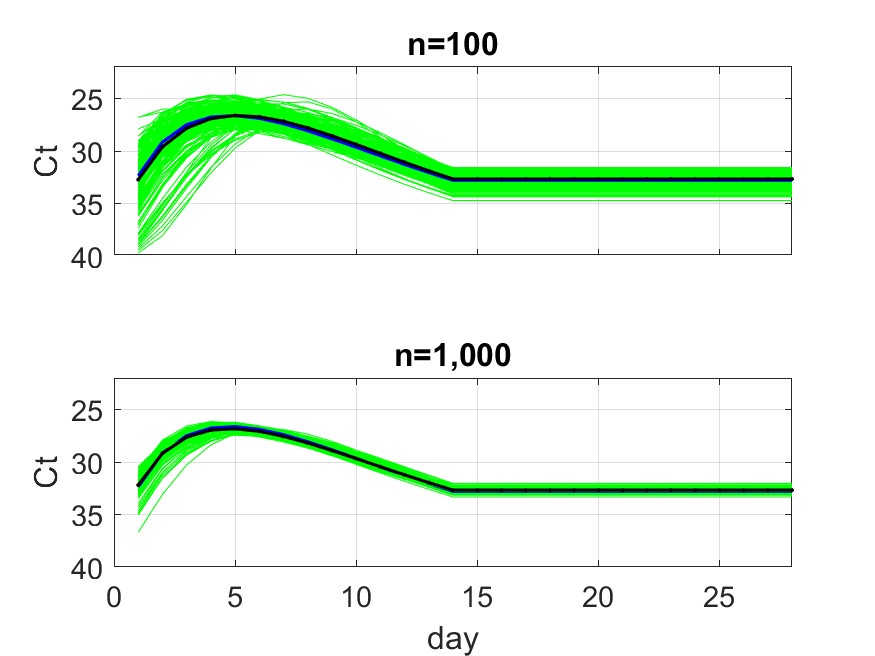}}
    \caption[]{Synthetic sample NMSE for the following estimates:   (a) $\ttheta$, (b) $\sigma,\rho$ in $\ssigma$  and (c) $\boldsymbol{q}$.
    (d): Plots of $\ttheta$ estimates of $200$ trials (green lines) for $n=100$ and $n=1,000$, together with their averages and  true $\ttheta$ (solid lines).}
    \label{fig:synth2}
  \end{minipage}
\end{figure}

The NMSE results versus the number of samples are shown in Figure \ref{fig:synth2}. 
Each NMSE value is calculated as the average of 200 trials.
Figures \ref{fig:synth2}(a), \ref{fig:synth2}(b), and \ref{fig:synth2}(c) show NMSE results for $\ttheta$, $(\sigma,\rho)$ and $\boldsymbol{q}$ respectively, in the three different settings of $\Delta$. Figure \ref{fig:synth2}(d) displays all estimates of $\ttheta$  from the $200$ trials for $n=100$ and $n=1,000$ (green lines). As expected, the variance is lower when $n=1,000$. The average of $200$ estimates coincide with the true $\ttheta$ (solid lines), which suggest the estimate is approximately unbiased.
Estimation of $\boldsymbol{q}$ is more challenging as the values have a smaller scale (between 0 and 1). Interestingly, the estimates of all parameters perform better when $\Delta\sim \text{Unif}(1,...,14)$.

Figure \ref{fig:synth3}(a) shows the semi-synthetic samples' NMSE results for all three parameters. The NMSE results were averaged over 100 trials, a sufficient number to produce smooth curves in this case.
Here $\ssigma$ has a linear covariance structure, so the NMSE is computed via the Frobenius norm. As previously, $\ttheta$ estimate is more accurate in terms of NMSE compared to $\boldsymbol{q}$, which is more difficult to estimate. Figure \ref{fig:synth3}(b) shows $200$ estimates of $\ttheta$ together with thier average and true $\ttheta$, for $n=100$ and $n=1,000$.

\begin{figure}[h]
  \begin{minipage}{\textwidth}
    \centering
    \subcaptionbox{NMSE of three parameters}
    {\includegraphics[width=0.45\textwidth]{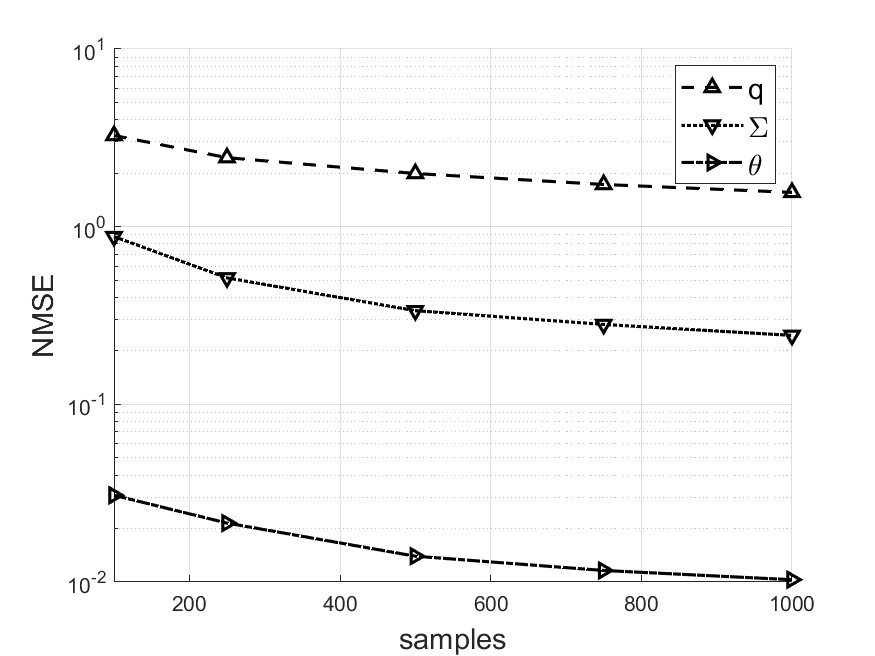}}\quad
    \subcaptionbox{plots of estimates}
     {\includegraphics[width=0.45\textwidth]{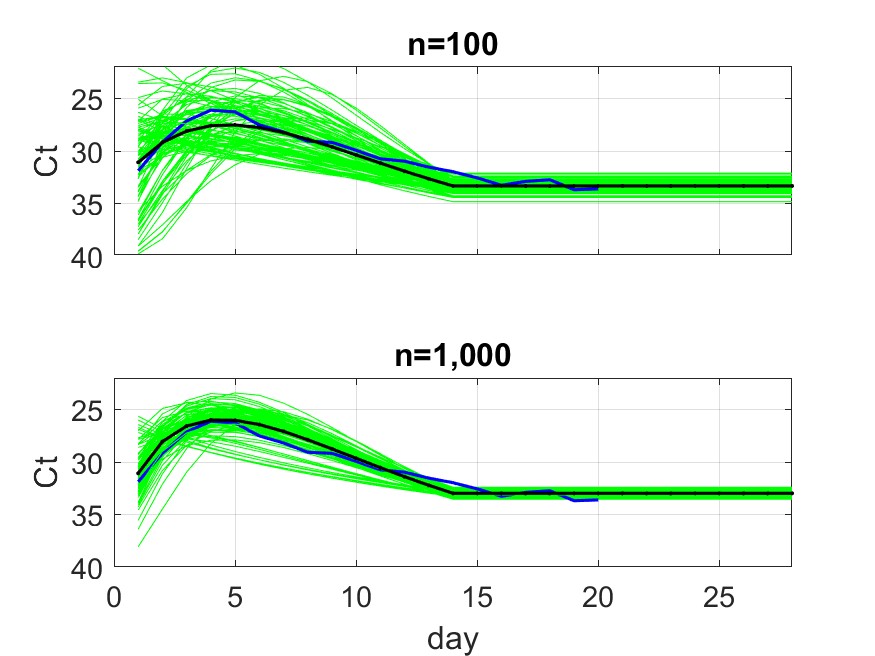}}
    \caption[]{(a): Semi-synthetic sample NMSE for the following  $\ttheta$,  $\sigma,\rho$ in $\ssigma$  and  $\boldsymbol{q}$. (b): Plots of $\ttheta$ estimates of $200$ trials (green lines) for $n=100$ and $n=1,000$, together with their averages and  true $\ttheta$ (solid lines).}
    \label{fig:synth3}
  \end{minipage}
\end{figure}

\subsection{Choice of $d$}
The latest day $d$ that limits the day after infection on which the first measurement was taken is somewhat arbitrary, reflecting a day on which the VL is low and individuals can be regarded as ``recovered". In this part, we examine how the choice of $d$ affects the estimate's accuracy. We use the same setting as in the first simulation model, using a uniform distribution of $\Delta$ over $\{1,...,14\}$. 
Generating a $100$ samples in each trial, we perform parameter estimation assuming various $d$ values: $d=7,10,14$ and $20$,  where $d=14$ is the ground-truth.

\begin{figure}[h]
  \begin{minipage}{\textwidth}
    \centering
    \subcaptionbox{Gamma}
    {\includegraphics[width=0.45\textwidth]{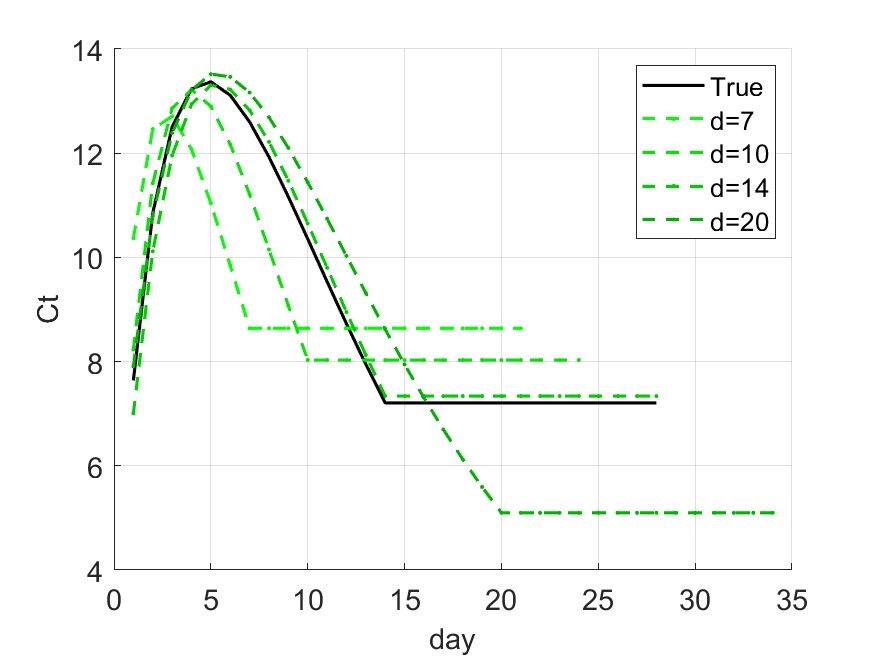}}\quad
    \subcaptionbox{Unimodal}
     {\includegraphics[width=0.45\textwidth]{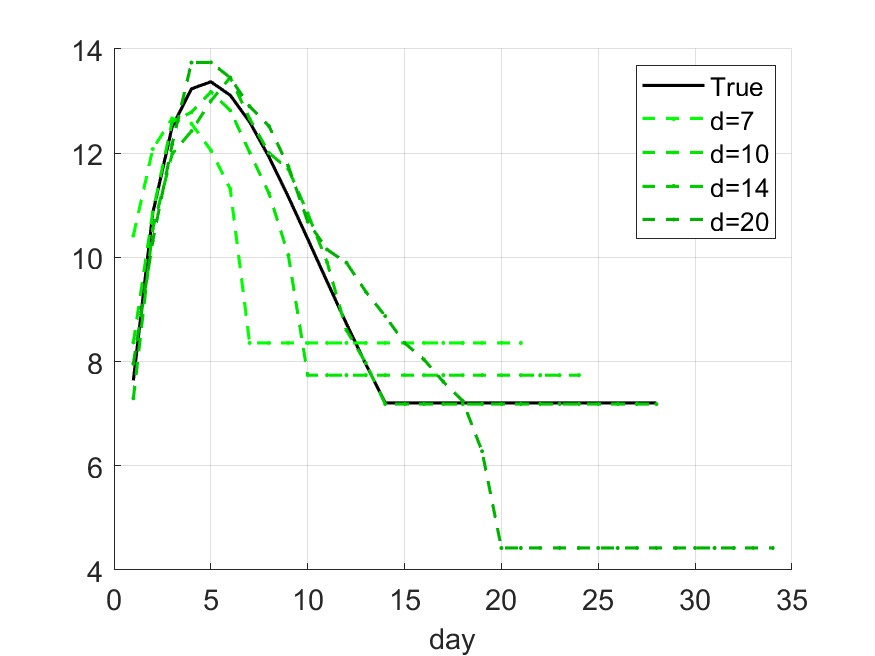}}
    \caption[]{(a): Plots of $\ttheta$ estimates
    for $d=7,10,14$ and $20$ assuming the unimodal Gamma function in \eqref{thetax}.
    (b): Plots of $\ttheta$ estimates assuming a general unimodal function.}
    \label{fig:choiced}
  \end{minipage}
\end{figure}

Panel (a) of Figure \ref{fig:choiced} shows the average results of the four estimates of $\ttheta$  compared to the ground-truth under model \eqref{thetax}, and panel (b) shows the same for a unimodal-constrained model.
The estimate coincide with the true $\ttheta$ in the early days, while towards the end the estimates differ. In general, the model is more sensitive to $d$ values smaller than the true $d$. The value on which the VL is maximal, which is an important parameter, seems quite stable to the value of $d$, especially under the unimodal constraint.

\section{Viral load data}

The data contains Nucleocapsid gene ($N$-gene) Ct-values, measured by a major lab in Israel, on swab samples taken from patients who were tested positive for SARS-Cov-2.
The samples were taken between January to June, 2022, during which the SARS-Cov-2 infections were most likely caused by the Omicron variant. 
The data contains records on patients whose Ct-value was measured once, twice or multiple times on different days. Specifically, out of $222,668$  records,  about $97\%$ records contain a single Ct-measurement,  while about  $2.7\%, 0.26\%$, and $0.04\%$  respectively contain  pairs, triples, and quadruples of Ct-values measured on different days. These amount to over $6,000$ pairs, $580$ triples, and $89$  quadruples of Ct-values.

Since the Ct-values are documented only on patients who were tested positive, the  $\hat{\ttheta}$ estimate should be interpreted as the daily mean Ct-value of all infected persons; more formally, the expectation of the Ct-value conditioned on Ct-value less than $40$. This is different from the mean Ct-value trajectory which includes recovered persons. The latter is usually the focus of viral load studies, such as  \cite{hay2022quantifying,chia2022virological}. Nevertheless,  because the probability of recovery (i.e., for the Ct-value to reach $40$) is low 
on the early days following infection,
these two trajectories are most likely to coincide on these  days. In any case, the last day $d$ should be chosen carefully and the estimate should be interpreted accordingly.

\begin{figure}[h]
  \begin{minipage}{\textwidth}
    \centering
    \subcaptionbox{$n=6,653$}
    {\includegraphics[width=0.45\textwidth]{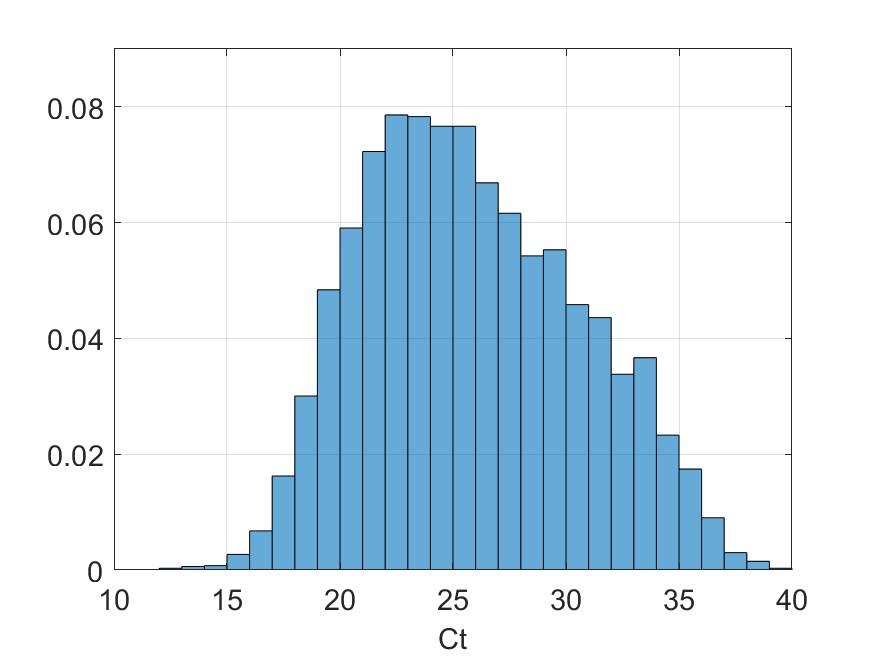}}\quad
    \subcaptionbox{$n=216,015$}
     {\includegraphics[width=0.45\textwidth]{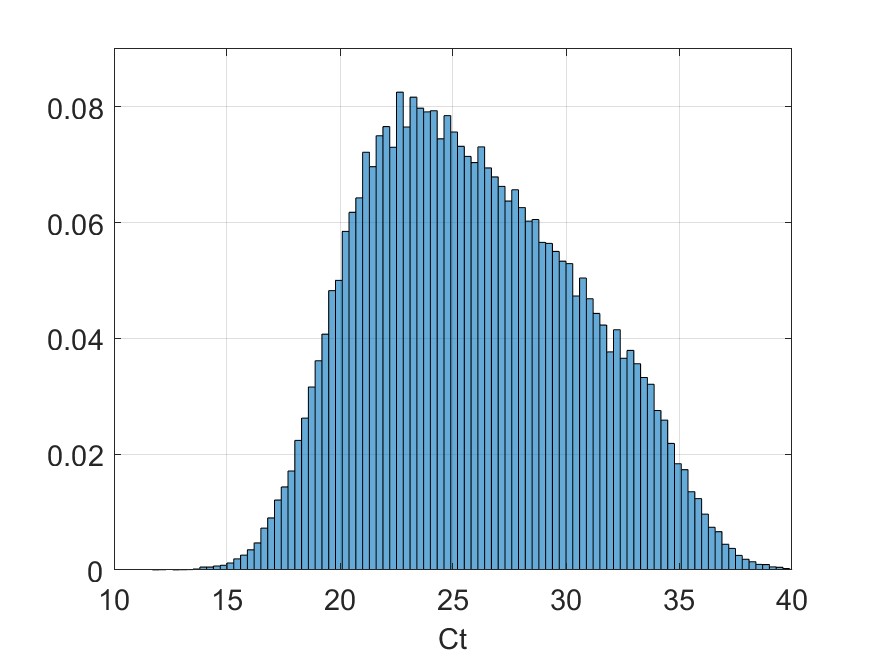}}
     \\
     \centering
    \subcaptionbox{$n=6,645$}
    {\includegraphics[width=0.45\textwidth]{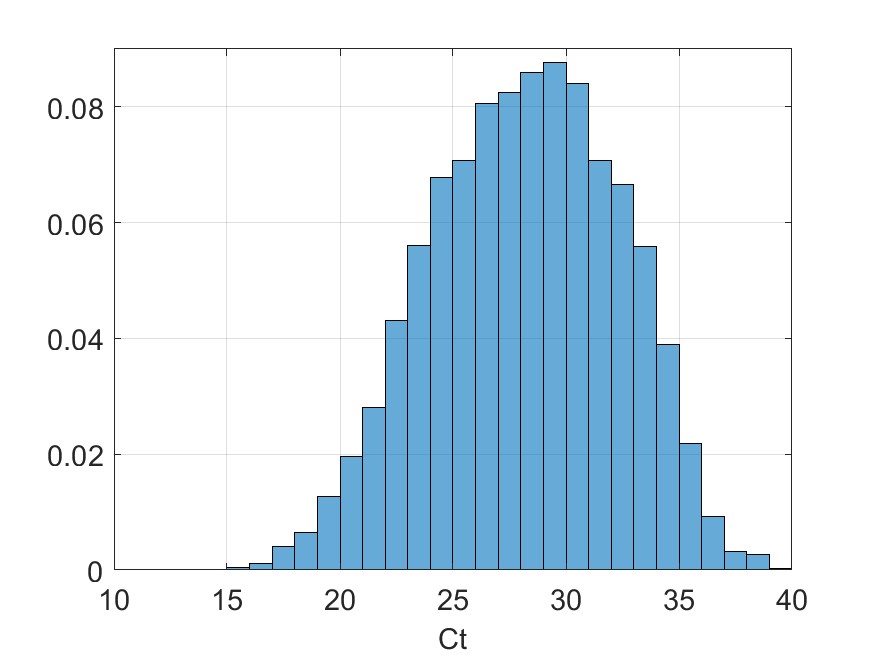}}\quad
    \subcaptionbox{$\Delta$}
     {\includegraphics[width=0.45\textwidth]{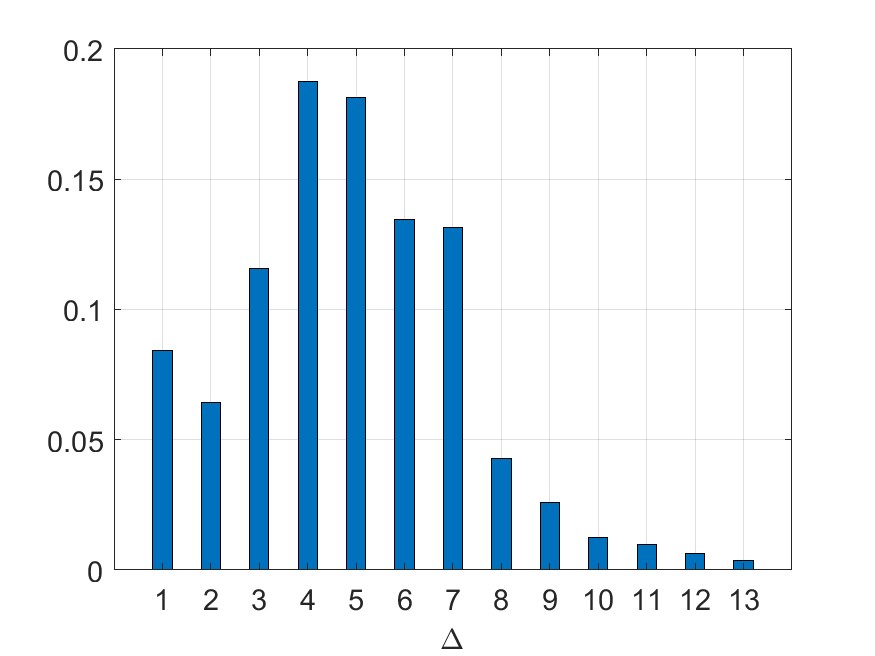}}
    \caption[]{Top row: distribution of first Ct-value measurement among: 
 persons who were tested multiple times (a), and persons who were tested once (b). Number of samples is given at the bottom.  (c): Distribution of second Ct-value measurement. (d):  Histogram of time difference $\Delta$ in days between the first and second measurement.}
    \label{fig:viral1}
  \end{minipage}
\end{figure}

Figure \ref{fig:viral1}(a) shows a histogram of the first Ct-value of individuals whose Ct-value was measured multiple times on different days (6,653 samples).  This histogram is similar to that of the Ct-value of individuals who were measured only once (Figure \ref{fig:viral1}(b), with 216,015 samples).
We also notice that although samples with Ct-value of $40$ are not included, still the tails decrease smoothly, without a sudden truncation at Ct-value of 40. 
Figure \ref{fig:viral1}(c) shows a histogram of the second Ct-value of individuals whose Ct-value was measured multiple times. The distribution is shifted more to the right, showing an average increase in Ct-value between the tests. As a larger Ct-value represents a smaller VL, this suggests that the second measurements were typically taken after the infection's peak. Figure \ref{fig:viral1}(d) shows the distribution of $\Delta$ - the difference in days between the first and  second measurements (samples with $\Delta>13$ were excluded; the number of such samples is negligible accounting for $1\%$ of all samples). Most of the $\Delta$ values are concentrated between days $1$ and $7$, which constitute about $90\%$ of all samples.

\begin{figure}[h]
  \begin{minipage}{\textwidth}
    \centering
    \subcaptionbox{Gamma}
    {\includegraphics[width=0.45\textwidth]{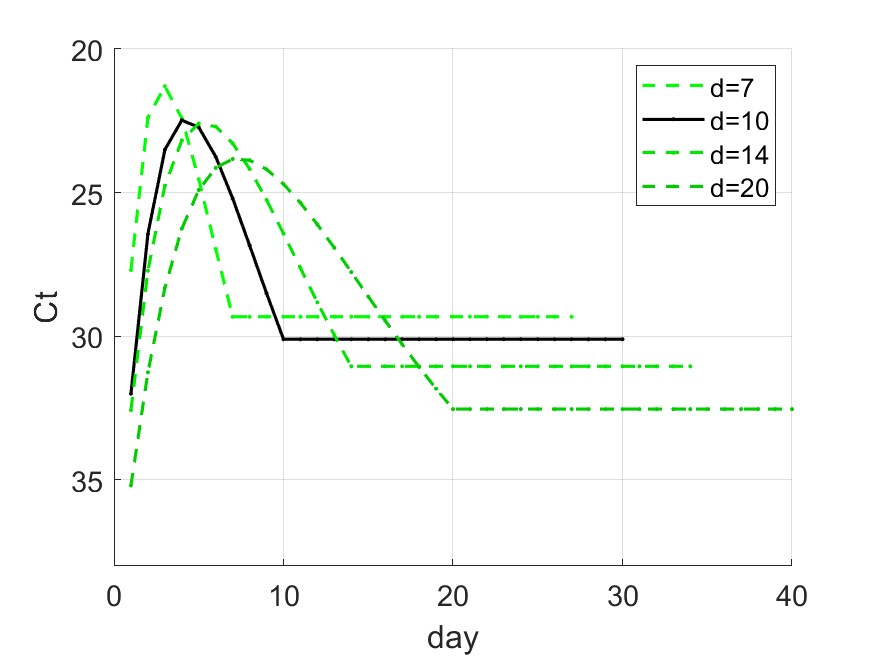}}\quad
    \subcaptionbox{Unimodal}
{\includegraphics[width=0.45\textwidth]{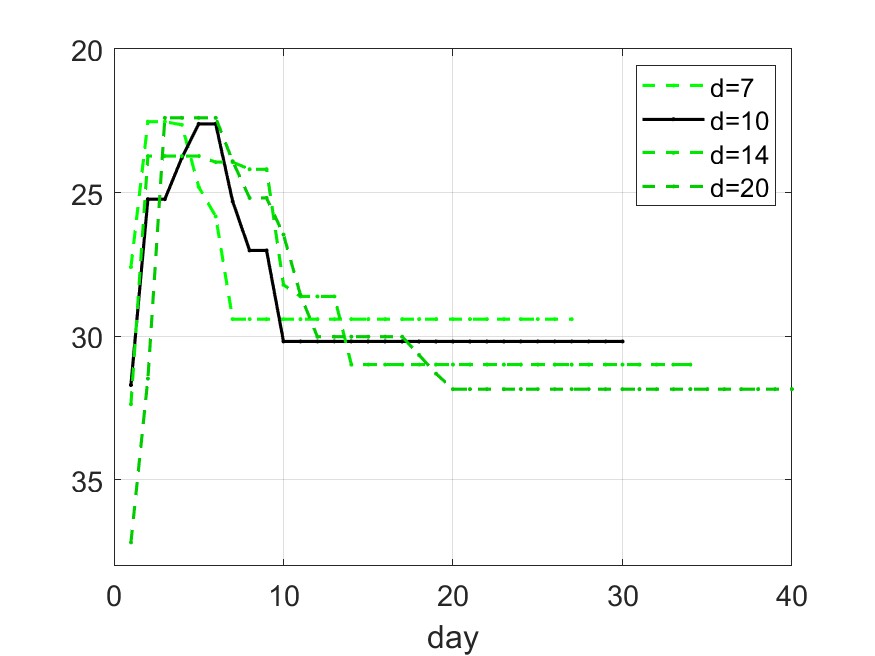}}
\caption[]{(a): Daily Ct-value estimates 
    for $d=7,10,14$ and $20$ based on the parametric model \eqref{thetax}.
    (b): 
    Daily Ct-value estimates for $d=7,10,14,20$ based on a unimodal constraint. 
    }
    \label{fig:viral2}
  \end{minipage}
\end{figure}

We included all pairs, triples, and quadruples of Ct-values in estimation.
We used the structural assumptions set in the semi-synthetic simulation, namely the Gamma parametric class of $\ttheta$ and the linear covariance structure (due to possible heteroscedasticity).
Although the maximum value of $\Delta$ is 13, the maximum number of days after infection, $d$, that best fits the data is unknown. A natural choice is $d=14$; however, we also examine $d=7,10$ and $20$. 
For each $d$, we follow the same procedure as described previously, running the EM algorithm several times using different random initializations, and selecting the one that gives the largest likelihood value.
Since the number of samples is quite large, we use $1,000$ iterations as convergence might possibly be slower.

Figure \ref{fig:viral2}(a) shows the daily Ct-value estimation results for $d=7,10,14$ and $20$. Since $d=10$ gives the largest likelihood value 
we regard this as the estimate of choice and it is shown as a black solid line.  The curve shows that the viral load reaches its peak during day 4  after infection (with Ct$\approx 23$), a result that is consistent with current knowledge (\cite{hay2022quantifying}). The Ct-value on day 10 reaches a value of 30, and it represents the average Ct-value on day 10 and onward. 
Panel (b) shows the daily Ct-value estimates under unimodality constraints for various $d$, which are quite similar to the parametric Gamma curves shown in (a).


\section{Discussion}
VL studies usually involve longitudinal samples with unknown time-origin. By developing and applying a likelihood-based EM algorithm, we reconstruct the daily mean SARS-Cov-2  VL starting from the day of infection, using real samples mostly composed of pairs of measurements taken from infected persons. The resulting estimate is consistent with previous studies \citep{hay2022quantifying} which use an exact time-origin.


In the midst of the COVID-19 pandemic, crucial policy decisions have been made under the pressure of limited evidence. One such example is determining the appropriate length of quarantine for infected individuals – a decision that carries substantial economic and ethical consequences. In order to ensure responsible and effective policies, it is imperative to base these decisions on rapid and accurate information.
Estimating viral load (VL) trajectories – a key factor in understanding infectiousness – traditionally requires resource-intensive and complex longitudinal studies. These studies are not only time-consuming but often involve substantial costs and logistical challenges. These led to delays in obtaining critical data needed for informed policy adjustments.

The method presented here requires much less resources, is logistically much simpler than traditional methods and can provide fast and important estimates for the parameters of interest. Importantly, the method's reliance on routine surveillance data makes it more adaptable and feasible for rapid implementation. In the face of emerging variants of concern, the ability to promptly evaluate data and adjust policies becomes even more vital. The current approach represents a significant advancement, offering the promise of faster decision-making processes while maintaining scientific rigor. Although the method is based on strong modelling assumptions, it can provide good initial estimates that can be updated when more complete data are collected.



In the synthetic simulation we examined how the distribution of the time difference $\Delta$ between pairs of samples affects the accuracy of the estimate. We found that the uniformly distributed $\Delta$ gave a stable and accurate estimate, which suggest what the optimal $\Delta$  should be. A detailed analysis which also accounts for all other parameters can strengthen these result and is a possible future research direction. 
Another important issue is the case where the samples greatly depart from the normal distribution. This includes, for instance, the bimodal distribution that can be relevant if recovered persons, whose VL becomes a fixed zero, are included in the samples.

\section*{Acknowledgements}
The authors would like to thank Arnona Ziv for providing the data.

\bibliographystyle{plainnat}
\bibliography{ref}

\section*{Appendix I: The E-step}
We need to calculate terms such as
\begin{equation} \label{expnorm}
\mathbbm{E}_{{\Theta}^{(t)}}\Big{(}(\y_i-\ttheta)^T \ssigma^{-1}
(\y_i-\ttheta) \mid x_{i,1}=j,D_i\Big{)}.
\end{equation}
Using  $(\y_i-\ttheta)^T \ssigma^{-1}
(\y_i-\ttheta) = {\rm trace}(\ssigma^{-1}(\y_i-\ttheta)(\y_i-\ttheta)^T)$,
and the additive property of the trace and expectation functionals, \eqref{expnorm} can be simplified to 
\begin{equation} \label{expsplit}
{\rm trace}\{\ssigma^{-1}\mathbbm{E}_{{\Theta}^{(t)}}(\y_i\y_i^T \mid x_{i,1}=j,D_i)\} -2{\rm trace}\{\ssigma^{-1}\ttheta\mathbbm{E}_{{\Theta}^{(t)}}(\y_i^T \mid x_{i,1}=j,D_i)\} +  \ttheta^T\ssigma^{-1}\ttheta.
\end{equation}
Let $\y_{i|j}^T=[(\y_{i|j}^{\rm obs})^T,(\y_{i|j}^{\rm unobs})^T]$ be partitioned into its observed and unobserved parts for the case $x_{i,1}=j$, that is, 
$\y_{i|j}^{\rm obs}=(y_{i,j},y_{i,j+\Delta_{i1}},\ldots,y_{i,j+\Delta_{im_i-1}})^T$ and $\y_{i|j}^{\rm unobs}$ is the rest of the vector. Similarly, let $(\ttheta^{(t)})^T=[(\ttheta_{i|j}^{(t){\rm obs}})^T,(\ttheta_{i|j}^{(t){\rm unobs}})^T]$ and 
$$
\ssigma^{(t)}=
\begin{pmatrix}
\ssigma_{i|j}^{(t){\rm obs}} & \ssigma_{i|j}^{(t){\rm obs,unobs}} \\ (\ssigma_{i|j}^{(t){\rm obs,unobs}})^T & \ssigma_{i|j}^{(t){\rm unobs}}
\end{pmatrix}
$$
be the corresponding partition of the current (after iteration $t$) estimates of the mean and variance functions of $\y_i$. The expectations in \eqref{expsplit} should be calculated as the expectation of $\y_i$ conditional on $\y_{i|j}^{\rm unobs}$ using the probability law:
\begin{equation} \label{distatt}
\y_{i|j} \equiv \begin{pmatrix} \y_{i|j}^{\rm obs} \\ \y_{i|j}^{\rm unobs} \end{pmatrix} \sim N \left(
\begin{pmatrix} \ttheta_{i|j}^{(t){\rm obs}} \\ \ttheta_{i|j}^{(t){\rm unobs}} \end{pmatrix},
\begin{pmatrix}
\ssigma_{i|j}^{(t){\rm obs}} & \ssigma_{i|j}^{(t){\rm obs,unobs}} \\ (\ssigma_{i|j}^{(t){\rm obs,unobs}})^T & \ssigma_{i|j}^{(t){\rm unobs}}
\end{pmatrix}\right) \; .
\end{equation}
Using properties of the multivariate normal distribution:
\begin{align}
\be_{i|j}^{(t)} &\equiv \mathbbm{E}_{{\Theta}^{(t)}}(
\y_{i|j}^{\rm unobs} \mid x_{i,1}=j, \y_{i|j}^{\rm obs})
    =
     \ttheta^{(t){\rm unobs}}_{i|j}+
\big{(}\ssigma^{(t){\rm obs,unobs}}_{i|j} \big{)}^T  
(\ssigma_{i|j}^{(t){\rm obs}})^{-1}
(\y_{i|j}^{\rm obs}-\theta_{i|j}^{(t){\rm obs}})
    \nonumber
 \\  \label{ecov23} \\
V_{i|j}^{(t)} &\equiv  \text{Var}_{{\Theta}^{(t)}}(
\y_{i|j}^{\rm unobs} \mid x_{i,1}=j, \y_{i|j}^{\rm obs})
    =
   \ssigma^{(t){\rm unobs}}_{i|j} - \big{(}\ssigma^{(t){\rm obs,unobs}}_{i|j} \big{)}^T (\ssigma_{i|j}^{(t){\rm obs}})^{-1} \big{(}\ssigma^{(t){\rm obs,unobs}}_{i|j} \big{)}^T
     \nonumber
\end{align}

From \eqref{ecov23} we obtain the terms for \eqref{expsplit}:
\begin{equation} \label{EYnew}
\y_{i|j}^{(t)}=\pi_{i|j} \circ \mathbbm{E}_{{\Theta}^{(t)}}(\y_{i|j} \mid x_{i,1}=j,D_i) = \pi_{i|j} \circ [(\y_{i|j}^{\rm obs})^T, (\be_{i|j}^{(t)} )^T]^T,
\end{equation}
and
\begin{equation} \label{EY2new}
\C_{i|j}^{(t)}=\pi_{i|j} \circ \mathbbm{E}_{{\Theta}^{(t)}}(\y_{i|j}\y_{i|j}^T \mid x_{i,1}=j,D_i) = \pi_{i|j} \circ
\begin{pmatrix}
\y^{\rm obs}_{i|j}(\y^{\rm obs}_{i|j})^T & \y^{\rm obs}_{i|j}(\be^{(t)}_{i|j})^T \\ \be^{(t)}_{i|j}(\y^{\rm obs}_{i|j})^T & V^{(t)}_{i|j} + \be^{(t)}_{i|j}(\be^{(t)}_{i|j})^T
\end{pmatrix},
\end{equation}
where $\pi_{i|j}$ rearrange the components to the original indexing.

\section*{Appendix II: Linear covariance implementation}
Consider a covariance matrix with a linear structure
$\ssigma(\boldsymbol{\beta}) = \sum_{j=1}^J \beta_j B_j$. That is, $\ssigma(\boldsymbol{\beta})$ is a linear combination of $J$ known symmetric matrices
$B_j$ ($j=1,\ldots,J$).
Recall that at the $t+1$-th iteration, after computing $\ttheta^{(t+1)}$, we have $\boldsymbol{Y}^{(t+1)}=\frac{1}{n}\sum\limits_{i=1}^n \sum\limits_{j=1}^{d-1} E_{ij}^{(t)} \boldsymbol{Y}^{(t)}_{\ttheta^{(t+1)},i|j}$; see \eqref{yy}-\eqref{ttheta}). Given the above  linear structure, our goal is to solve:
\begin{equation*}
    \min_{\boldsymbol{\beta}} \text{log}|\ssigma|
    +\text{Tr} \Bigg{(}
    \ssigma^{-1} \boldsymbol{Y}^{(t+1)}
    \Bigg{)}\quad 
    s.t. \quad \ssigma=\sum\limits_{j=1}^p \beta_j \boldsymbol{B}_{j}.
\end{equation*}
Note that we isolated all terms which depend on $\ssigma$ in \eqref{eloglc}, and divided by $\sum\limits_{i=1}^n \sum\limits_{j=1}^{d-1} E_{ij}^{(t)}=n$. 
Following \cite{anderson1973asymptotically},  the minimum point is found by repeatedly solving the linear equation $\boldsymbol{G}^{(\ell)}\boldsymbol{\beta}^{(\ell+1)}=\boldsymbol{g}^{(\ell)}$, where:
\begin{align*}
    \boldsymbol{G}^{(\ell)}\in \mathbbm{R}^{J \times J}: \quad [\boldsymbol{G}^{(\ell)}]_{k_1 k_2}
    &=
    \text{Tr}
    \Bigg{[}
    \bigg{(}
    \sum\limits_{j=1}^J
    \beta^{(\ell)}_j 
    \boldsymbol{B}_j
    \bigg{)}
    \boldsymbol{B}_{k_1}
    \bigg{(}
    \sum\limits_{j=1}^J
    \beta^{(\ell)}_j 
    \boldsymbol{B}_j
    \bigg{)}
    \boldsymbol{B}_{k_2}
    \Bigg{]} \\
    \boldsymbol{g}^{(\ell)}\in \mathbbm{R}^{J}: \quad [\boldsymbol{g}^{(\ell)}]_{k_1 }
    &=
    \text{Tr}
    \Bigg{[}
    \bigg{(}
    \sum\limits_{j=1}^J
    \beta^{(\ell)}_j 
    \boldsymbol{B}_j
    \bigg{)}
    \boldsymbol{B}_{k_1}
    \bigg{(}
    \sum\limits_{j=1}^J
    \beta^{(\ell)}_j 
    \boldsymbol{B}_j
    \bigg{)}
    \boldsymbol{Y}^{(t+1)}
    \Bigg{]},
\end{align*} 
initialized by $\boldsymbol{\beta}^{(0)}$, until convergence. The final estimate $\hat{\boldsymbol{\beta}}$ forms the covariance $\hat{\ssigma}=\sum\limits_{j=1}^J \hat{\beta}_j\boldsymbol{B}_j$.

\end{document}